\documentclass[english,aps,prl,twocolumn,groupedaddress,superscriptaddress,nofootinbib]{revtex4-1}
\usepackage{amsfonts}
\usepackage{amssymb}
\usepackage{amsmath}
\usepackage{color}
\usepackage{dcolumn}
\usepackage{bm}
\usepackage{babel}
\usepackage{amsmath}
\usepackage{graphicx}
\usepackage{amsfonts}
\usepackage{amssymb}
\usepackage{color}
\usepackage{mathrsfs}
\usepackage{isomath}
\usepackage{amsmath}
\usepackage{amsthm}
\usepackage{epstopdf}
\usepackage{txfonts}
\usepackage{float}
\usepackage{dsfont}

\usepackage[all]{xy}
\usepackage[colorinlistoftodos]{todonotes}
\usepackage[colorlinks=true, linkcolor=blue, citecolor=red]{hyperref}

\newcommand{\Tr}{\mbox{Tr}}

\def\be{\begin{equation}}
\def\ee{\end{equation}}
\def\ba{\begin{array}}
\def\ea{\end{array}}
\def\Tr{\mathrm{Tr}}

\newtheorem{thm}{Theorem}
\newtheorem{lem}[thm]{Lemma}
\newtheorem{cor}[thm]{Corollary}

\newtheorem{definition}{Definition}

\newtheorem{example}{Example}

\newcommand{\p}{\mathbf{p}}
\newcommand{\q}{\mathbf{q}}
\newcommand{\ab}{\mathbf{a}}
\newcommand{\bb}{\mathbf{b}}
\newcommand{\cb}{\mathbf{c}}
\newcommand{\db}{\mathbf{d}}
\newcommand{\rrr}{\mathbf{r}}
\newcommand{\sss}{\mathbf{s}}
\newcommand{\ttt}{\mathbf{t}}
\newcommand{\x}{\mathbf{x}}
\newcommand{\y}{\mathbf{y}}

\makeatletter
\newcommand*\bigcdot{\mathpalette\bigcdot@{.5}}
\newcommand*\bigcdot@[2]{\mathbin{\vcenter{\hbox{\scalebox{#2}{$\m@th#1\bullet$}}}}}
\makeatother

\usepackage{tikz}
\usetikzlibrary{patterns}
\makeatletter
\def\grd@save@target#1{%
  \def\grd@target{#1}}
\def\grd@save@start#1{%
  \def\grd@start{#1}}
\tikzset{
  grid with coordinates/.style={
    to path={%
      \pgfextra{%
        \edef\grd@@target{(\tikztotarget)}%
        \tikz@scan@one@point\grd@save@target\grd@@target\relax
        \edef\grd@@start{(\tikztostart)}%
        \tikz@scan@one@point\grd@save@start\grd@@start\relax
        \draw[minor help lines,magenta] (\tikztostart) grid (\tikztotarget);
        \draw[major help lines] (\tikztostart) grid (\tikztotarget);
        \grd@start
        \pgfmathsetmacro{\grd@xa}{\the\pgf@x/1cm}
        \pgfmathsetmacro{\grd@ya}{\the\pgf@y/1cm}
        \grd@target
        \pgfmathsetmacro{\grd@xb}{\the\pgf@x/1cm}
        \pgfmathsetmacro{\grd@yb}{\the\pgf@y/1cm}
        \pgfmathsetmacro{\grd@xc}{\grd@xa + \pgfkeysvalueof{/tikz/grid with coordinates/major step}}
        \pgfmathsetmacro{\grd@yc}{\grd@ya + \pgfkeysvalueof{/tikz/grid with coordinates/major step}}
        \foreach \x in {\grd@xa,\grd@xc,...,\grd@xb}
        \node[anchor=north] at (\x,\grd@ya) {\pgfmathprintnumber{\x}};
        \foreach \y in {\grd@ya,\grd@yc,...,\grd@yb}
        \node[anchor=east] at (\grd@xa,\y) {\pgfmathprintnumber{\y}};
      }
    }
  },
  minor help lines/.style={
    help lines,
    step=\pgfkeysvalueof{/tikz/grid with coordinates/minor step}
  },
  major help lines/.style={
    help lines,
    line width=\pgfkeysvalueof{/tikz/grid with coordinates/major line width},
    step=\pgfkeysvalueof{/tikz/grid with coordinates/major step}
  },
  grid with coordinates/.cd,
  minor step/.initial=.2,
  major step/.initial=1,
  major line width/.initial=2pt,
}
\makeatother

\usetikzlibrary{arrows, decorations.markings}

\tikzstyle{vecArrow} = [thick, decoration={markings,mark=at position
   1 with {\arrow[semithick]{open triangle 60}}},
   double distance=1.4pt, shorten >= 5.5pt,
   preaction = {decorate},
   postaction = {draw,line width=1.4pt, fill=white, opacity=0,shorten >= 4.5pt}]
\tikzstyle{innerWhite} = [semithick, white,line width=1.4pt, shorten >= 4.5pt]

\begin{document}

\title{$\mbox{The Uncertainty Principle of Quantum Processes}$}

\author{Yunlong Xiao}
\email{mathxiao123@gmail.com}
\affiliation{School of Physical and Mathematical Sciences, Nanyang Technological University, Singapore 639673, Singapore}
\affiliation{Complexity Institute, Nanyang Technological University, Singapore 639673, Singapore}
\affiliation{Department of Mathematics and Statistics, University of Calgary, Calgary, Alberta T2N 1N4, Canada}
\affiliation{Institute for Quantum Science and Technology, University of Calgary, Calgary, Alberta, T2N 1N4, Canada}

\author{Kuntal Sengupta}
\email{kuntal.sengupta.in@gmail.com}
\affiliation{Department of Mathematics and Statistics, University of Calgary, Calgary, Alberta T2N 1N4, Canada}
\affiliation{Institute for Quantum Science and Technology, University of Calgary, Calgary, Alberta, T2N 1N4, Canada}

\author{Siren Yang}
\affiliation{Department of Mathematics and Statistics, University of Calgary, Calgary, Alberta T2N 1N4, Canada}
\affiliation{Institute for Quantum Science and Technology, University of Calgary, Calgary, Alberta, T2N 1N4, Canada}

\author{Gilad~Gour}
\email{gour@ucalgary.ca}
\affiliation{Department of Mathematics and Statistics, University of Calgary, Calgary, Alberta T2N 1N4, Canada}
\affiliation{Institute for Quantum Science and Technology, University of Calgary, Calgary, Alberta, T2N 1N4, Canada}

%%%%%%%%%%%%%%%%%%%%%%%%%%%%%%%%%%%

\begin{abstract}
Heisenberg's uncertainty principle, which imposes intrinsic restrictions on our ability to predict the outcomes of incompatible quantum measurements to arbitrary precision, demonstrates one of the key differences between classical and quantum mechanics. The physical systems considered in the uncertainty principle are static in nature and described mathematically with a quantum state in a Hilbert space. However, many physical systems are dynamic in nature and described with the formalism of a quantum channel.
In this paper, we show that the uncertainty principle can be reformulated to include process-measurements that are performed on quantum channels. Since both quantum states and quantum measurements are themselves special cases of quantum channels, our formalism encapsulates the uncertainty principle in its utmost generality. More specifically, we obtain expressions that generalize the Maassen-Uffink uncertainty relation and the universal uncertainty relations from quantum states to quantum channels.
\end{abstract}

%%%%%%%%%%%%%%%%%%%%%%%%%%%%%%%%%%%

\maketitle

%%%%%%%%%%%%%%%%%%%%%%%%%%%%%%%%%%%

\textit{Introduction.-} 
Counter-intuitive as it may seem, the uncertainty principle has been firmly rooted as a fundamental restriction that lies in the heart of quantum mechanics \cite{Heisenberg1927}. The amount of information one can extract from a quantum system, at any given time, depends on the extent of the incompatibility of the underlying measurements involved. In the Hiesenberg's uncertainty principle, this corresponds to the fact that any attempt to measure the position of a quantum particle with very high precision, comes at a cost of poor precision in the simultaneous measurement of its momentum. This fundamental distinction from classical physics has led to an enormous research in the area and found a plethora of applications in quantum key distribution~\cite{berta2010heisenberg}, and the detection of quantum resources~\cite{RevModPhys.91.025001}, such as entanglement~\cite{PhysRevA.68.032103,PhysRevA.68.034307,PhysRevLett.92.117903,PhysRevA.70.022316,PhysRevLett.119.170404,PhysRevLett.122.220401}, Einstein- Podolsky-Rosen steering~\cite{PhysRevA.40.913,PhysRevA.87.062103,PhysRevLett.118.020402,PhysRevA.97.052307,xiao2018quasifinegrained,PhysRevA.98.050104}, and Bell nonlocality~\cite{Oppenheim1072}. 

The uniqueness and immense potential in quantum uncertainty has caused Heisenberg's uncertainty principle -- which was mathematically formulated by Kennard~\cite{Kennard1927} (also refer to Weyl~\cite{weyl1928gruppentheorie}) -- to go through multiple refinements over the last century. One such example is the use of R\'enyi entropies to formulate uncertainty relations from an information-theoretic perspective, by Maassen and Uffink \cite{PhysRevLett.60.1103} (based on the Riesz theorem \cite{hardy1952inequalities}):
\begin{align}\label{mu}
\mathrm{H}_{\alpha} (M) + \mathrm{H}_{\beta} (N)
\geqslant -2\log c \left( M , N \right).
\end{align}
Here $\mathrm{H}_{\alpha} (M) := \frac{1}{1-\alpha}\log (\sum_{x} p_{x}^{\alpha})$ stands for the R\'enyi entropy with order $\alpha>0$, where $\p=\{p_x\}$ is the probability vector corresponding to the outcomes of the measurement $M$, when performed on a system in a state $\rho$. The R\'enyi parameters $\alpha$ and $\beta$ are chosen such that $1/\alpha + 1/\beta = 2$. The constant $c \left( M , N \right)$ stands for the maximal overlap between the measurements $M$ and $N$, and is independent on the state $\rho$. 

More recently, the authors of \cite{PhysRevLett.111.230401} showed that not only entropic functions, but any non-negative Schur-concave function is a suitable uncertainty quantifier for the probabilities obtained from measurements, giving rise to a class of infinitely many uncertainty relations, namely universal uncertainty relations (UURs). It is critical to note that the classification of UURs with respect to the joint uncertainty they present is a major focus in the theory of uncertainty relations \cite{yuan2019strong}. In particular, while considering the two probability distributions $\p$ and $\q$ obtained by measuring quantum state $\rho$ with respect to measurements $M$ and $N$, their joint uncertainty based on direct-product \cite{PhysRevLett.111.230401,Pucha_a_2013},
\begin{align}\label{uur}
\p\otimes\q &\prec \bb_{\otimes},
\end{align}
 demonstrates a spatially-separated type of uncertainties \cite{yuan2019strong}. Here $\bb_{\otimes}$ is a probability vector independent of the initial state $\rho$, and hence quantifies the inherent incompatibility between the measurements and ``$\prec$'' stands for majorization (For $\x = (x_{k})_{k}$, $\y = (y_{k})_{k} \in \mathds{R}^{d}$, we have $\x \prec \y$ whenever $\sum_{k=1}^{i} x_{k}^{\downarrow} \leqslant \sum_{k=1}^{i} y_{k}^{\downarrow}$ for all $1 \leqslant i \leqslant d-1$ and $\sum_{k=1}^{d} x_{k} = \sum_{k=1}^{d} y_{k}$, where the downarrow $^{\downarrow}$ means the components of corresponding vector are arranged in non-increasing order). The approach of majorization adopted by \cite{PhysRevLett.111.230401,Pucha_a_2013} frees us from particular measures and captures the essence of uncertainty in quantum mechanics. Alternatively, investigators have also sought to express the joint uncertainty through direct-sum \cite{PhysRevA.89.052115}
\begin{align}\label{uur2}
\p\oplus\q &\prec \bb_{\oplus},
\end{align}
which reveals a temporally-separated joint uncertainty \cite{yuan2019strong}, where the vector $\bb_{\oplus}$ is also independent of the initial state.

The tendency of quantum states is to evolve. In any realizable experimental setup, all state preserving and transforming operations come with an error. For example, a state preserving experimental apparatus whose only job is to retain the state of the quantum system over a short time, mathematically doing nothing to the quantum state, always accounts for its tolerance, precision and the least count(s) of the measuring device(s) associated with it. Therefore, errors are simply inevitable in any experiment. This error is the device-dependent unavoidable evolution of a quantum system : a quantum process. A complete formulation of the uncertainty principle, therefore, must involve the study of uncertainty relations for quantum processes. Moreover, since quantum processes generalize quantum states and measurements on states, this study will allow us to extend the notion of quantum uncertainty to a theory of quantum mechanics modelled solely by quantum processes.

This paper is organized as follows. First, we give a brief introduction of terminologies and background information on the process positive-operator-valued measure, i.e. PPOVM, which will be useful throughout. The very first step of quantifying uncertainty of a quantum state is to perform incompatible measurements on it and extract sets of probability distributions corresponding to those measurements. Therefore, we next describe what it means to measure a quantum channel. Secondly, we propose a R\'enyi entropic uncertainty relation followed by direct-sum and direct-product UURs for quantum processes. The bounds of all these three relations are independent of the process at hand. Moreover, our results generalize the celebrated uncertainty relations (\ref{mu}), (\ref{uur}), (\ref{uur2}) for quantum states to quantum processes, and the extension of our result to multiple process-channel measurements is straightforward. Examples to support our result are also provided. Finally, we address a number of interesting directions of future investigations which have a close connection to the framework explored here. 

%%%%%%%%%%%%%%%%%%%%%%%%%%%%%%%%%%%

\textit{Preliminaries.-}
For a finite dimensional Hilbert space $\mathcal{H}$, the set of all linear transformation taking the Hilbert space to itself is denoted by $L (\mathcal{H})$. An operator $\rho \in L (\mathcal{H})$ is a density operator, representing a quantum state, if it is  positive semi-definite and has unit trace, i.e., $\rho \geqslant 0$ and $\Tr [\rho] = 1$. We denote the collection of all density operators on $\mathcal{H}$ as $D (\mathcal{H})$. A quantum effect $M_{x}$ on $D (\mathcal{H})$ is an operator such that $0 \leqslant M_{x} \leqslant \mathds{1}$, where, $\mathds{1}$ denotes the identity matrix on $D (\mathcal{H})$. The probability $p_x$ of an outcome $x$ as a result of the effect $M_x$ acting on a density operator $\rho$ is given by $\Tr [M_{x} \, \rho]$. A positive-operator-valued measure (POVM) $M = \{ M_{x} \}_{x}$ is a set of effects that collectively sum to the identity $\sum_{x} M_{x} = \mathds{1}$. As one can easily see that this makes $\{p_x\}_x$ a true set of probabilities. It is important to point out here that although a density operator fully characterizes the statistical properties of the corresponding quantum state at a given time, it is with the help of the measurements that information contained in the state can be retrieved. 

 A superoperator $\mathrm{\Psi}$ maps operators of one Hilbert space to operators of another Hilbert space. Naturally, for Hilbert spaces $\mathcal{H}^A$ and $\mathcal{H}^B$,  $\mathrm{\Psi} (L (\mathcal{H}^{A})) \subset L (\mathcal{H}^{B})$ or simply $\mathrm{\Psi} : A \rightarrow B$, and their collection is denoted by $T( A , B ) := \{ \mathrm{\Psi} \, | \, \mathrm{\Psi} : A \rightarrow B \}$. A superoperator $\mathrm{\Psi} \in T( A , B )$ is said to be a quantum channel if it is (\romannumeral1) completely positive (CP, i.e., $\mathds{1}^{R} \otimes \mathrm{\Psi}$ is positive for all finite dimensional Hilbert space $\mathcal{H}^{R}$), and (\romannumeral2) trace-preserving (TP, i.e. $\Tr_{B} [ \mathrm{\Psi} ( \bigcdot ) ] = \Tr_{A} [ \bigcdot ]$). It is crucial to note here that a quantum channel is a special quantum process which preserves probabilities, a deterministic process, and in this work we are only interested in studying the indeterministic phenomena associated with incompatible measurements on deterministic processes, thereby truly capturing the essence of uncertainty. We use $\text{CPTP}( A , B )$ to denote the collection of all CPTP maps from space $L (\mathcal{H}^{A})$ to space $L (\mathcal{H}^{B})$. 
 
 Now let us look at states and measurements from the perspective of quantum channels : a quantum state $\rho$ can be seen as the state-preparation channel $\mathrm{\Gamma}_{\rho}: \mathds{C} \rightarrow \mathcal{H}$, and a POVM $M = \{ M_{x}\}_{x}$ is equivalent to the measurement channel $\mathrm{\Lambda}_{M}: \mathcal{H} \rightarrow \mathds{C}$. In order to deal with quantum channels, i.e., CPTP maps, it is convenient to use the Choi-Jamio\l kowski isomorphism \cite{JAMIOLKOWSKI1972275,CHOI1975285} :

\begin{lem}[Choi-Jamio\l kowski]
For any $\mathrm{\Psi} \in T( A , B )$ , there is a linear bijection between $T( A , B )$ and $L (A \otimes B) := L (\mathcal{H}^{A} \otimes \mathcal{H}^{B})$, which is given by
\begin{align}
\theta : T( A , B ) &\rightarrow L (A \otimes B) \notag \\
\mathrm{\Psi} & \mapsto
J_{\mathrm{\Psi}}^{A B} 
\end{align}
where $J_{\mathrm{\Psi}}^{A B} = \mathds{1}^{\tilde{A} \rightarrow A} \otimes \mathds{1}^{B} (J_{\mathrm{\Psi}}^{\tilde{A} B}) = \mathds{1}^{\tilde{A} \rightarrow A} \otimes \mathrm{\Psi} ( \phi_{+}^{\tilde{A} A} )$ with $\phi_{+}^{\tilde{A} A} := | \phi_{+}^{\tilde{A} A} \rangle \langle \phi_{+}^{\tilde{A} A} |$ being an unnormalized maximally entangled state with $| \phi_{+}^{\tilde{A} A} \rangle := \sum_{i=1}^{d_{A}} | i \rangle ^{\tilde{A}} | i \rangle ^{A}$. 
\end{lem}

\noindent Here the tilde symbol indicates an identical copy of the system under it, and hence in what follows, we do not distinguish between the space $\mathcal{H}^{\tilde{A}}$ and $\mathcal{H}^{A}$. Analogous to the role of density operator of a quantum system, the Choi-Jamio\l kowski matrix $J_{\mathrm{\Psi}}^{A B}$, or briefly CJ matrix, provides a complete description of the physical process $\mathrm{\Psi}$ at any given time. 

The entire study of quantum information theory revolves around how much information can be efficiently packed, transferred and retrieved in a desired fashion by means of preparation, manipulation and measurement of quantum states \cite{nielsen_chuang_2010,wilde_2017,watrous_2018}. In the theory of quantum mechanics modelled solely by quantum channels, this whole picture, therefore, boils down to the idea of storing and retrieving information, not from quantum states, but from quantum channels themselves. However, just like quantum states, the only way of accessing information from a quantum channel is by measuring it. Such a measurement is called process-channel measurement, first introduced in \cite{PhysRevA.77.062112}. Imagine a scenario with a state-preparation device providing initial state $\rho^{R A} \in L (R \otimes A)$ and a POVM $M = \{ M_{x} \}_{x}$ acting on $L (R \otimes B)$. Formally, the process-channel measurement is defined by the couple $\mathcal{T}:=( \rho^{R A}, M )$. To retrieve the information conveyed by the quantum channel $\mathrm{\Psi}$, a reference system $R$ distributes to the measuring device directly, meanwhile the probe system $A$, which is correlated with $R$, is transformed through $\mathrm{\Psi}$ followed by the measurement $M$. Here the classical information or measurement outcome $x$ will occur with probability $p_{x} = \Tr [ M_{x}  \mathds{1}^{R} \otimes \mathrm{\Psi} ( \rho^{R A} ) ]$. Next we substitute the equation $\rho^{R A} = \mathrm{\Upsilon}_{\rho} \otimes \mathds{1}^{A} ( \phi_{+}^{\tilde{A} A} )$, where $\mathrm{\Upsilon}_{\rho} : L ( \mathcal{H}^{\tilde{A}} ) \rightarrow L ( \mathcal{H}^{R} )$ is a CP linear map \cite{wilde_2017,watrous_2018}, into the expression for $p_{x}$ we get 
\begin{align}
p_{x} &= \Tr \left[ M_{x}  \mathds{1}^{R} \otimes \mathrm{\Psi} \left( \rho^{R A} \right) \right]\notag\\
&= \Tr \left[ M_{x}  \mathds{1}^{R} \otimes \mathrm{\Psi} 
\left( \mathrm{\Upsilon}_{\rho} \otimes \mathds{1}^{A} \left( \phi_{+}^{\tilde{A} A} \right)  \right) \right]\notag\\
&= \Tr \left[ \mathrm{\Upsilon}_{\rho}^{\ast} \otimes \mathds{1}^{B} \left( M_{x} \right)  \mathds{1}^{A} \otimes \mathrm{\Psi} 
\left( \phi_{+}^{\tilde{A} A} \right) \right]\notag\\
&= \Tr \left[ \mathrm{\Upsilon}_{\rho}^{\ast} \otimes \mathds{1}^{B} \left( M_{x} \right)  J_{\mathrm{\Psi}}^{A B} \right]
\end{align}
where in the third equation, $\mathrm{\Upsilon}_{\rho}^{\ast}$ is the dual map of $\mathrm{\Upsilon}_{\rho}$, which is also CP linear map, with the property that for all operators $M^{A} \in L(A) := L ( \mathcal{H}^{A} )$ and for all $M^{R} \in L(R) := L ( \mathcal{H}^{R} )$, we have $\Tr [ ( \mathrm{\Upsilon}_{\rho} ( M^{A} )  )^{\dagger}  M^{R} ] = \Tr [ ( M^{A} )^{\dagger}  \mathrm{\Upsilon}_{\rho}^{\ast} ( M^{R} ) ]$. 

In regard to above discussions, it is clear that for each single channel measurement $( \rho^{R A}, M_{x} )$, we can define an operator $E_{x} := \mathrm{\Upsilon}_{\rho}^{\ast} \otimes \mathds{1}^{A} ( M_{x} ) \geqslant 0$ satisfying 
\begin{align}\label{pro}
p_{x} = \Tr \left[ E_{x}  J_{\mathrm{\Psi}}^{A B} \right]
\end{align}
Here $E_{x}$ is the so-called process-channel effect of single channel measurement $( \rho^{R A}, M_{x} )$, and their collection $\{ E_{x} \}_{x}$ is known as process POVM (PPVOM) or tester \cite{PhysRevA.77.062112}. More generally, a PPOVM is a special case of $2$-comb \cite{Chiribella_2008,PhysRevLett.101.060401,8678741}, with pre-processing and post-processing are classical-to-quantum and quantum-to-classical channels respectively.

%%%%%%%%%%%%%%%%%%%%%%%%%%%%%%%%%%%

\textit{Maassen-Uffink Uncertainty Relations.-} 
Having defined what a measurement of quantum process is, we now use it to study entropic uncertainty relations. Let $\mathrm{\Psi}$ be a quantum channel from operator space $L(A)$ to $L(B)$. For simplicity of the exposition, we start with two PPOVMs, and denote them as $\mathcal{T}_{1}:=( \rho^{R A}, M )$ and $\mathcal{T}_{2}:=( \sigma^{R A}, N )$. We also denote by $\{p_{x}\}_x$ and $\{q_{y}\}_y$ the two probability distributions obtained by measuring $\mathrm{\Psi}$ with respect to $\mathcal{T}_{1}$ and $\mathcal{T}_{2}$.
In analogy with $E_{x}$ and $p_{x}$, let us also define $F_{y} = \mathrm{\Upsilon}_{\sigma}^{\ast} \otimes \mathds{1}^{A} ( N_{y} ) \geqslant 0$ as the process-channel effect of single channel measurement $( \sigma^{R A}, N_{y} )$, such that, $q_{y} = \Tr [ F_{y}  J_{\mathrm{\Psi}}^{A B} ]$. It is straight forward to check that $\sum_{x} E_{x} = (\rho^{A})^{\mathrm{T}} \otimes \mathds{1}^{B} \leqslant  \mathds{1}^{A B}$ and  $\sum_{y} F_{y} = (\sigma^{A})^{\mathrm{T}} \otimes \mathds{1}^{B} \leqslant  \mathds{1}^{A B}$, where $^{\mathrm{T}}$ denotes transposition in the corresponding space, and hence the mathematical structure of PPOVMs do not obey the completeness relation \cite{PhysRevA.77.062112}. 

Next we will introduce the overlap for PPOVMs by extending the sets of process-channel effect $\{ E_{x} \}_{x=1}^{m}$ and $\{ F_{y} \}_{y=1}^{n}$ to $\{ \tilde{E}_{x} \}_{x=1}^{m+1}$ and $\{ \tilde{F}_{y} \}_{y=1}^{n+1}$, respectively. In regards to the subscript, the extended process-channel effects $\tilde{E}_{x}$, $\tilde{F}_{y}$ are defined as :
\begin{equation}
\tilde{E}_{x} := \left\{
\begin{aligned}
&E_{x}  ~ & 1 \leqslant x \leqslant m, \\
\mathds{1}^{A B} - &(\rho^{A})^{\mathrm{T}} \otimes \mathds{1}^{B}  ~ & x=m+1.
\end{aligned}
\right.
\end{equation}
and
\begin{equation}
\tilde{F}_{y} := \left\{
\begin{aligned}
&F_{y}  ~ & 1 \leqslant y \leqslant n, \\
\mathds{1}^{A B} - &(\sigma^{A})^{\mathrm{T}} \otimes \mathds{1}^{B}  ~ & y=n+1.
\end{aligned}
\right.
\end{equation}
The quantity $c_{xy} ( \mathcal{T}_{1} , \mathcal{T}_{2} ) := \Vert \tilde{E}_{x}^{1/2}  \tilde{F}_{y}^{1/2} \Vert$ with $1 \leqslant x \leqslant m+1$ and $1 \leqslant y \leqslant n+1$, represents the overlap between process-channel measurements $\mathcal{T}_{1}$ and $\mathcal{T}_{2}$, analogous to the overlap between projective measurements \cite{PhysRevLett.50.631}, extensively investigated in many quantum information-theory contexts, for example, \cite{PhysRevLett.50.631}. The maximum overlap between  $\mathcal{T}_{1}$ and $\mathcal{T}_{2}$, then can be defined as $c ( \mathcal{T}_{1} , \mathcal{T}_{2} ) := \max_{x, y} c_{xy} ( \mathcal{T}_{1} , \mathcal{T}_{2} )$. Guided by intuition, $c$ should provide a bound on the minimum uncertainty arising from simultaneously measuring $\mathrm{\Psi}$ with $\mathcal{T}_{1}$ and $\mathcal{T}_{2}$, thereby quantifying the inherent incompatibility between process-channel measurements. We will return to the meaning of this later.

\begin{figure}[h]
\centering
\begin{tikzpicture}

    % left plot
    \node[] at (-2.91,-0.4) {(a) Process POVM (PPOVM) $\mathcal{T}_{1}$};
    \draw[very thick,dashed,fill=magenta,opacity=0.2] (-5,3) rectangle (-0.75,0); 
    \node[] at (-4.7,2.7) {$\mathcal{T}_{1}$};
    \node[] at (-4.6,1.475) {$\rho^{RA}$};
    \draw[very thick,->] (-4.3,1.475) -- (-3.8,2.15) -- (-2,2.15);
    \draw[very thick,->] (-4.3,1.475) -- (-3.8,0.8) -- (-3.2,0.8);
    \node[] at (-2.9,2.4) {$\mathcal H_{R}$};
    \node[] at (-3.5,0.55) {$\mathcal H_{A}$};
    \draw[very thick,->] (-2.6,0.8) -- (-2,0.8);
    \node[] at (-2.3,0.55) {$\mathcal H_{B}$};
    \draw[very thick,dashed,fill={rgb:black,1;white,2}] (-3.2,1.2) rectangle (-2.6,0.5);
    \node[] at (-2.9,0.85) {\large $\mathrm{\Psi}$};
    
    %left measurement
    \draw[very thick] (-2,2.45) rectangle (-1.2,0.5);
    \draw[very thick] (-1.3,1.3) arc (10:170:0.3);
    \draw[very thick,->] (-1.6,1.3) -- (-1.4,1.75);
    \node[] at (-1.6,2) {$M$};
    %\draw[very thick,vecArrow] (-1.2,1.475) -- (-0.8,1.475);
    \draw[very thick,-] (-1.2,1.525) -- (-0.9,1.525) -- (-0.9,1.625) -- (-0.8,1.475) 
    -- (-0.9,1.325) -- (-0.9,1.425) -- (-1.2,1.425);
    \node[] at (-1.05,1.7) {$\p$};
    
    % right plot (+4.5)
    \node[] at (1.59,-0.4) {(b) Process POVM (PPOVM) $\mathcal{T}_{2}$};
    \draw[very thick,dashed,fill=cyan,opacity=0.2] (-0.5,3) rectangle (3.75,0); 
    \node[] at (-0.2,2.7) {$\mathcal{T}_{2}$};
    \node[] at (-0.1,1.475) {$\sigma^{RA}$};
    \draw[very thick,->] (0.2,1.475) -- (0.7,2.15) -- (2.5,2.15);
    \draw[very thick,->] (0.2,1.475) -- (0.7,0.8) -- (1.3,0.8);
    \node[] at (1.6,2.4) {$\mathcal H_{R}$};
    \node[] at (1,0.55) {$\mathcal H_{A}$};
    \draw[very thick,->] (1.9,0.8) -- (2.5,0.8);
    \node[] at (2.2,0.55) {$\mathcal H_{B}$};
    \draw[very thick,dashed,fill={rgb:black,1;white,2}] (1.3,1.2) rectangle (1.9,0.5);
    \node[] at (1.6,0.85) {\large $\mathrm{\Psi}$};
    
    %right measurement
    \draw[very thick] (2.5,2.45) rectangle (3.3,0.5);
    \draw[very thick] (3.2,1.3) arc (10:170:0.3);
    \draw[very thick,->] (2.9,1.3) -- (3.1,1.75);
    \node[] at (2.9,2) {$N$};
    \draw[very thick,-] (3.3,1.525) -- (3.6,1.525) -- (3.6,1.625) -- (3.7,1.475) 
    -- (3.6,1.325) -- (3.6,1.425) -- (3.3,1.425);
    \node[] at (3.45,1.7) {$\q$};

\end{tikzpicture}
\caption{(color online) Schematic illustration of the PPOVMs.}
\label{ppovm}
\end{figure}
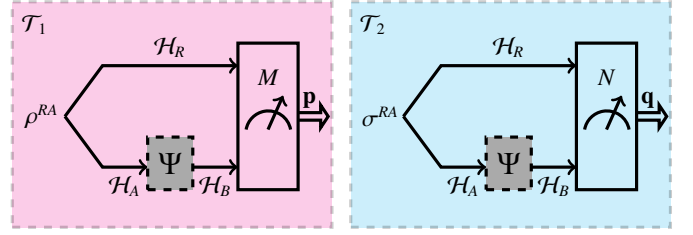

To establish our uncertainty relations, we collect the probabilities $p_{x}$ and $q_{y}$ into two probability vectors $\p := \p ( \mathcal{T}_{1} , \mathrm{\Psi} ) = ( p_1, \ldots, p_m )$ and $\q := \q ( \mathcal{T}_{2} , \mathrm{\Psi} ) = ( q_1, \ldots, q_n )$ respectively (Fig. \ref{ppovm}). Having defined what the
probability vectors for process-channel measurements are, we now consider the corresponding uncertainty measure.

In classical information theory, entropy describes the uncertainty associated with a random variable, and hence becomes a suitable candidate for uncertainty measure \cite{PhysRevLett.50.631}. Inspired  by Maassen and Uffink \cite{PhysRevLett.60.1103}, we base our first result on the class of R\'enyi entropies defined as :
\begin{align}
\mathrm{H}_{\alpha} \left( \p \right) := - \frac{1}{1-\alpha} \log \left(\sum\limits_{x=1}^{m} p_{x}^{\alpha}\right),
\end{align}
with $\alpha>0$ and $\alpha \neq 1$. 

As with the primal formulation of Maassen-Uffink uncertainty relations \cite{PhysRevLett.60.1103}, our first goal is to bound the joint uncertainty between $\p$ and $\q$ in terms of $\mathrm{H}_{\alpha} (\mathcal{T}_{1}) + \mathrm{H}_{\beta} (\mathcal{T}_{2})$. This will turn out to be different from ~\cite{Gao2018UncertaintyPF}, by a quantity which depends only on the process effects $E_{x}$ and $F_{y}$ but not on the process $\mathrm{\Psi}$ itself.

\begin{thm}\label{muthm}
For probability vectors $\p$ and $\q$ obtained by measuring $\mathrm{\Psi}$ with respect to $\mathcal{T}_{1}:=( \rho^{R A}, M )$ and $\mathcal{T}_{2}:=( \sigma^{R A}, N )$, their joint uncertainties in terms of $\mathrm{H}_{\alpha} (\mathcal{T}_{1}) + \mathrm{H}_{\beta} (\mathcal{T}_{2})$ is bounded by the maximum overlap $c ( \mathcal{T}_{1} , \mathcal{T}_{2} )$ as :
\begin{align} 
\label{mustrong}
\mathrm{H}_{\alpha} ( \frac{1}{d_{A}} \p \oplus \frac{d_{A}-1}{d_{A}} )
+
\mathrm{H}_{\beta} ( \frac{1}{d_{A}} \q \oplus \frac{d_{A}-1}{d_{A}} ) \geqslant 
-2\log c ( \mathcal{T}_{1} , \mathcal{T}_{2} ),
\end{align}
where $\alpha$ and $\beta$ satisfy the harmonic condition $1/\alpha + 1/\beta = 2$. 
\end{thm}

\noindent The left hand side of Eq. (\ref{mustrong}) relies on the initial state of the quantum process, $\mathrm{\Psi}$, and the incompatible process-channel measurements $\mathcal{T}_{1}$, $\mathcal{T}_{2}$. The right-hand side is an irreducible bound for its joint uncertainty, which depends only on the incompatible process-channel measurements $\mathcal{T}_{1}$, $\mathcal{T}_{2}$, and can be calculated explicitly. With the help of Eq. (\ref{mustrong}), we can also derive the Shannon entropic uncertainty relation for quantum processes by taking limits of $\alpha$ and $\beta$ to approach 1. 

It is interesting to remark that for a state-preparation channel $\mathrm{\Gamma}_{\rho}$, we have $d_{A} = \dim \mathds{C} = 1$ for any state $\rho$, and hence, $\frac{1}{d_{A}} \p \oplus \frac{d_{A}-1}{d_{A}} = \p$, $\frac{1}{d_{A}} \q \oplus \frac{d_{A}-1}{d_{A}} = \q$. Additionally, the PPOVMs will degenerate into POVMs, and the maximum overlap $c ( \mathcal{T}_{1} , \mathcal{T}_{2} )$ reduces to $c \left( M , N \right)$ \cite{Rastegin_2010}. This shows why our theorem \ref{muthm} includes Maassen-Uffink uncertainty relation as a special case. Moreover, note that the bound described by Eq (\ref{mustrong}) is tight, since for the case with $\alpha = \infty$ or $\beta = \infty$, $-2\log c ( \mathcal{T}_{1} , \mathcal{T}_{2} )$ is achieved by some quantum process \cite{xiao2019complementary}.

%%%%%%%%%%%%%%%%%%%%%%%%%%%%%%%%%%%

\textit{Universal Uncertainty Relations.-} 
We now turn our attention to the universal uncertainty relations for quantum processes. Traditionally, entropies like $\mathrm{H}_{\alpha}$ have been employed to study the uncertainty of probability distribution associated with measurements. However, in \cite{PhysRevLett.111.230401}, the authors showed that the notion of  majorization can fully characterize the uncertainty related with probability distributions and therefore capture the ``the essence of uncertainty in quantum mechanics''. Another motivation for the considerations of majorization uncertainty relations is that majorization, as a preorder, is more informative than the ones based on particular uncertainty measures, such as Shannon or R\'enyi entropies. Here, we will see that the joint distributions $\p \otimes \q$ and $\p \oplus \q$ obtained by measuring quantum processes are bounded by vectors independent of $\mathrm{\Psi}$.

Let us first collect all process effects from $\mathcal{T}_{1}$, $\mathcal{T}_{2}$ together, and define their collections as
\begin{equation}
G_{z} := \left\{
\begin{aligned}
&E_{z}  \quad &~ 1 \leqslant &z \leqslant m, \\
&F_{z-m}  \quad &~ m+1 \leqslant &z \leqslant m+n.
\end{aligned}
\right.
\end{equation}
It follows that the general experiments measuring the quantum process $\mathrm{\Psi}$ with $\mathcal{T}_{1}$ and $\mathcal{T}_{2}$ are completely characterized by the set of process effects $G$. For a subset $\mathcal{I}_{k} \subset \left\{1, \ldots, m+n \right\}$ with cardinality $k$, we define $G(\mathcal{I}_{k}) := \sum_{z \in \mathcal{I}_{k}} G_{z}$. With these conventions, the second goal of our work is to bound the joint uncertainty in the form of $\p \oplus \q$. More precisely,

\begin{thm}\label{uuroplusthm}
For probability vectors $\p$ and $\q$ obtained by measuring $\mathrm{\Psi}$ with respect to $\mathcal{T}_{1}:=( \rho^{R A}, M )$ and $\mathcal{T}_{2}:=( \sigma^{R A}, N )$, their joint uncertainties in terms of $\p \oplus \q$ is bounded by a vector independent of quantum process $\mathrm{\Psi}$ of the form
\begin{align}\label{uuroplus}
\p\oplus\q &\prec  
\sss 
:=
\left( s_{1} , s_{2} - s_{1} , s_{3} - s_{2} , \ldots, 0 \right),
\end{align}
where each $s_{k}$ is a functional of the conditional min-entropy
\begin{align}\label{sk}
s_{k} := \max\limits_{\mathcal{I}_{k}} 2^{ -\mathrm{H}_{\min} (B|A)_{ G\left( \mathcal{I}_{k} \right) } },
\end{align}
and the maximization is over all subsets $\mathcal{I}_{k}$. The conditional min-entropy for $G\left( \mathcal{I}_{k} \right)$ is defined as
\begin{align}
\mathrm{H}_{\min} (B|A)_{G\left(\mathcal{I}_{k}\right)} := - \log \inf\limits_{X^{A} \geqslant 0} \left\{ \Tr \left(X^{A} \right) | X^{A} \otimes \mathds{1}^B \geqslant G\left(\mathcal{I}_{k}\right) \right\}.
\end{align}
\end{thm}

\noindent Note that the operator $G\left(\mathcal{I}_{k}\right)$ is a process-channel effect, which is also a unnormalized quantum state. Thus, the conditional min-entropy defined above is not our usually used one for bipartite states. In order to find a formula based on our conversant conditional min-entropy, we can define a bipartite quantum state as $\tau \left(\mathcal{I}_{k}\right) = G\left(\mathcal{I}_{k}\right) / \Tr \left[ G\left(\mathcal{I}_{k}\right) \right]\in D \left( A \otimes B \right)$, which depends on the subset $\mathcal{I}_{k}$, and call it process-channel state corresponding to $G\left(\mathcal{I}_{k}\right)$. Consequently, $\mathrm{H}_{\min} (B|A)_{ \tau^{A B} \left( \mathcal{I}_{k} \right) }$ is conditional min-entropy of the bipartite state $\tau^{A B} \left( \mathcal{I}_{k} \right)$. Now the quantity $s_{k}$ can be expressed as:
\begin{align}
s_{k} = \max\limits_{\mathcal{I}_{k}} 2^{ \left( -\mathrm{H}_{\min} (B|A)_{ \tau \left( \mathcal{I}_{k} \right) } + 
\log \Tr \left[ G\left(\mathcal{I}_{k}\right) \right] \right)}.
\end{align}
We remark that the tightness of $s_{k}$ and the rigorous proof of  Thm. \ref{uuroplusthm} are detailed in the Supplemental Material. 

Aside from its numerous applications in single-shot quantum information, quantum hypothesis testing, and quantum resource theories, we show that this entropic quantifier has operational significance in terms of the tightness of the UURs for quantum processes with direct-sum form, which might also have an impact on the development of future technologies of quantum processes.

Continuing our discussion of UURs for quantum processes, we now show that the joint uncertainty based on direct product, i.e., $\p\otimes\q$, can also be similarly characterized as stated in Thm \ref{uurotimesthm}.

\begin{thm}\label{uurotimesthm}
For probability vectors $\p$ and $\q$ obtained by measuring $\mathrm{\Psi}$ with respect to $\mathcal{T}_{1}:=( \rho^{R A}, M )$ and $\mathcal{T}_{2}:=( \sigma^{R A}, N )$, their joint uncertainties in terms of $\p \otimes \q$ is therefore bounded by a vector independent of quantum process $\mathrm{\Psi}$ of the form
\begin{align}\label{uurotimes}
\p\otimes\q &\prec  
\ttt
:=
\left( t_{1}, t_{2}-t_{1}, t_{3}-t_{2}, \ldots, 0 \right),
\end{align}
with $t_{k}$ is defined by $ (s_{k+1}/2)^{2}$ constructed in Thm. \ref{uuroplusthm}.
\end{thm}

We finish by remarking that the class of Schur-concave functions can preserve the pre-order induced by majorization; that is for a Schur-concave function $\mathrm{\Phi} : \mathds{R}^{d} \rightarrow \mathds{R}$ and $x$, $y \in \mathds{R}^{d}$, $\mathrm{\Phi} (\x) \geqslant \mathrm{\Phi} (\y)$ whenever $\x \prec \y$. As a result, the UURs  for quantum processes in terms of $\p\oplus\q \prec \sss$ and $\p\otimes\q \prec \ttt$ generate an infinite family of uncertainty relations of the forms $\mathrm{\Phi}(\p\oplus\q) \geqslant \mathrm{\Phi}(\sss)$ and $\mathrm{\Phi}(\p\otimes\q) \geqslant \mathrm{\Phi}(\ttt)$ with each $\mathrm{\Phi}$. Taking $\mathrm{\Phi}$ as Shannon entropy $\mathrm{H}$, (\ref{uuroplus}), (\ref{uurotimes}) will lead to the Shannon entropic uncertainty relations for quantum processes $\mathrm{H} ( \mathcal{T}_{1} ) + \mathrm{H} ( \mathcal{T}_{2} ) \geqslant \mathrm{H} ( \sss )$ and $\mathrm{H} ( \mathcal{T}_{1} ) + \mathrm{H} ( \mathcal{T}_{2} ) \geqslant \mathrm{H} ( \ttt )$ with $\mathrm{H} ( \mathcal{T}_{1} ) := \mathrm{H} ( \p )$ and $\mathrm{H} ( \mathcal{T}_{2} ) := \mathrm{H} ( \q )$. However, the result presented in Thm. \ref{muthm} are not covered by UURs, since in (\ref{mustrong}) the uncertainty associated with $\p$ and $\q$ are quantified by different uncertainty measures.

%%%%%%%%%%%%%%%%%%%%%%%%%%%%%%%%%%%

\textit{Conclusions and Discussions.-} 
In this work we have addressed the question of whether quantum mechanics will obstruct us from predicting the outcomes of incompatible process-channel measurements to arbitrary precision. We studied uncertainty relations in three distinct forms: Maassen-Uffink form; direct-sum form; and direct-product form, which reduces to the well-known Maassen-Uffink entropic uncertainty relations \cite{PhysRevLett.60.1103} and UURs \cite{PhysRevLett.111.230401,Pucha_a_2013,PhysRevA.89.052115} as our special cases by choosing the process $\mathrm{\Psi}$ to be a state-preparation channel $\mathrm{\Gamma}_{\rho}$, i.e. $\mathrm{\Psi} = \mathrm{\Gamma}_{\rho}$. 

In particular, following Deutsch's observation \cite{PhysRevLett.50.631}, in order to express the uncertainty principle for quantum processes $\mathrm{\Psi} : A \rightarrow B$ quantitatively, we are seeking an inequality with the form
$
\mathcal{U} ( \mathcal{T}_{1}, \mathcal{T}_{2}, \mathrm{\Psi} )
\geqslant
\mathcal{B} ( \mathcal{T}_{1}, \mathcal{T}_{2} ),
$
where the quantity on the left-hand side represents the joint probability distribution induced by measuring quantum process $\mathrm{\Psi}$ with PPOVMs $\mathcal{T}_{1}$ and $\mathcal{T}_{1}$ in the form of $\mathcal{U}$, with the optimal bound $\mathcal{B} \left( \mathcal{T}_{1}, \mathcal{T}_{2} \right) := \min_{ \mathrm{\Psi} \in \text{CPTP}( A , B )} \mathcal{U} \left( \mathcal{T}_{1}, \mathcal{T}_{2}, \mathrm{\Psi} \right)$. If we denote the set of all state-preparation channels as $\mathrm{\Gamma} \subset \text{CPTP}( A , B )$, the celebrated Heisenberg's uncertainty principle, with the form
$
\mathcal{U} ( \mathcal{T}_{1}, \mathcal{T}_{2}, \mathrm{\Gamma}_{\rho} )
\geqslant
\min_{ \mathrm{\Psi} \in \mathrm{\Gamma}} \mathcal{U} 
( \mathcal{T}_{1}, \mathcal{T}_{2}, \mathrm{\Psi} ),
$
becomes a special case of our generalized uncertainty principle.

Our first main result shows that the potential knowledge one can have about any quantum process from pair of process-channel measurements, $\mathcal{T}_1$ and $\mathcal{T}_2$, quantified by the R\'enyi entropies with harmonic condition, is restricted by their inherent incompatibility in terms of $c ( \mathcal{T}_{1} , \mathcal{T}_{2} )$. Moreover, in our upcoming work of experimental investigations of uncertainty principle for quantum processes performed in a photonic system, we will show that (\ref{mustrong}) is tight.

Secondly, we derived the UURs for quantum processes, i.e. (\ref{uuroplus}) and (\ref{uurotimes}), which are the generalizations of the previous ones for quantum states, and are explicitly computable. A natural question is whether the process-independent bounds $\sss$ and $\ttt$ are optimal. For the sum of each $k$ distinct elements in $\p \oplus \q$, their upper-bound $s_{k}$ is tight, which means $s_{k}$ is achieved by performing $\mathcal{T}_{1}$ and $\mathcal{T}_{2}$ to some quantum processes. However, the vector $\sss$ consists of $s_{k}$ is not optimal. In the Supplemental Material \cite{sm}, we show that the optimal bound for $\p \oplus \q$ exists and is given by the vector $\mathcal{F} (\sss)$ with $\mathcal{F}$ stands for the flatness process \cite{992785}. On the other hand, even though the existence of optimal bound $\rrr$ for $\p \otimes \q$ is guaranteed by the completeness of majorization lattice \cite{BAPAT199159,BONDAR1994115,Bosyk_2019,doi:10.1002/andp.201900143}, so far we do not have any effective method in calculating it in general. Although the bound $\ttt$ introduced in (\ref{uurotimes}) is weaker when compares with $\rrr$, it is easy-to-evaluate. Similar to the method for direct-sum, the flatness process $\mathcal{F}$ can further improve the bound of direct-product to $\p \otimes \q \prec \rrr \prec \mathcal{F} (\ttt) \prec \ttt$. As a by-product of UURs for quantum processes, we show that the optimal bound for direct-sum form is specified completely by the conditional min-entropy, which connects UURs with single-shot information theory.

There are a plenty of important directions of investigations which we leave for future work. First of all, we did not explore here the extension of our results to the cases with bipartite quantum channels \cite{gour2019entanglement,buml2019resource}, where the measured quantum channel is prepared entangled with another channel, a dynamic quantum memory that might be possible to predict the outcomes for both process-channel measurements $\mathcal{T}_{1}$, $\mathcal{T}_{2}$ simultaneously, which is the generalized uncertainty principle in the presence of dynamic quantum memory \cite{berta2010heisenberg}. It would be also interesting to study how the use of dynamic quantum memory can further strengthen the power of quantum cryptography.

Another important direction of investigation is the noise and disturbance tradeoff in process-channel measurements \cite{PhysRevLett.112.050401}. To capture the idea of ``how accurate'' a process-channel measurement $\mathcal{T}_{1}$ is, we should consider its measuring apparatus $\mathscr{T}$, and the corresponding error $\mathscr{E} ( \mathcal{T}_{1} , \mathscr{T} )$, or noise, which is quantifies through a operational measurement statistics. When the measured channel is subjected to the apparatus $\mathscr{T}$, another process-channel measurement $\mathcal{T}_{2}$ will be disturbed and lead to the disturbance $\mathscr{D} ( \mathcal{T}_{2} , \mathscr{T} )$. The aim of this direction of investigation is to introduce the operational definitions for $\mathscr{E}$ and $\mathscr{D}$ such that $\mathscr{E} ( \mathcal{T}_{1} , \mathscr{T} ) +\mathscr{D} ( \mathcal{T}_{2} , \mathscr{T} ) \geqslant -2\log c ( \mathcal{T}_{1} , \mathcal{T}_{2} )$.

Finally, when considering the process-channel measurements with possibilities of small errors, we should employ smooth entropies to obtain meaningful results. Therefore, it would be important to generalize our entropic uncertainty relation for quantum processes to the one expressed in terms of smooth entropies \cite{PhysRevLett.106.110506}. Nevertheless, these generalizations are nontrivial and are left for future work.

%%%%%%%%%%%%%%%%%%%%%%%%%%%%%%%%%%%

\begin{acknowledgements}
We would like to thank Eric Chitambar, Kun Fang, Li Gao, Mile Gu, Anna Jen\v{c}ov\'a, Nicholas LaRacuente, Zhihao Ma, Varun Narasimhachar, Carlo Maria Scandolo, Gaurav Saxena, Jayne Thompson, Kunkun Wang, Peng Xue, Lei Xiao, and Yuxiang Yang for fruitful discussions. Y. X., and G. G. acknowledge financial support from the Natural Sciences and Engineering Research Council of Canada (NSERC).\end{acknowledgements}

%%%%%%%%%%%%%%%%%%%%%%%%%%%%%%%%%%%

\bibliographystyle{apsrev4-1}
\bibliography{YXiaoBib.bib}
%%%%%%%%%%%%%%%%%%%%%%%%%%%%%%%%%%%

\appendix
{
\begin{center}
{\bfseries Supplemental Material}
\end{center}
}
\setcounter{equation}{0}

%%%%%%%%%%%%%%%%%%%%%%%%%%%%%%%%%%%

\subsection{Proof of Theorem \ref{muthm}}

In this section we turn our attention to the Maassen-Uffink-form uncertainty relations for quantum processes. We will first briefly review the historical developments of Maassen-Uffink uncertainty relation, before formulating our generalized uncertainty principle in terms of R\'enyi entropies.

In 1983, Deutsch first introduced the uncertainty principle in terms of Shannon entropy for any two non-degenerate observables \cite{PhysRevLett.50.631}. The improved bound on Deutsch uncertainty relation was conjectured by Kraus in 1987 \cite{PhysRevD.35.3070}, and was proved by Maassen and Uffink one year later \cite{PhysRevLett.60.1103}. The uncertainty measure adopted by Maassen and Uffink is R\'enyi entropy, an improvement over Shannon entropic uncertainty relations. The original result of \cite{PhysRevLett.50.631} is only valid for pure states with Von Neumann measurements, and their proof relies on Riesz theorem \cite{hardy1952inequalities}. It is thus natural to ask whether Maassen-Uffink uncertainty relation also holds for mixed states with POVMs, which was shown to be correct by Rastegin in 2010 \cite{Rastegin_2010}. 

\begin{lem}[Rastegin]\label{mixedmulemma}
For probability vectors $\p$ and $\q$ obtained by measuring quantum state $\rho$ with respect to POVMs $M$ and $N$, their joint uncertainties in terms of $\mathrm{H}_{\alpha} (M) + \mathrm{H}_{\beta} (N)$ is therefore bounded by the maximum overlap $c ( M , N , \rho )$ of the form
\begin{align} 
\mathrm{H}_{\alpha} ( M )
+
\mathrm{H}_{\beta} ( N ) \geqslant 
-2\log c ( M , N , \rho ),
\end{align}
where $\alpha$ and $\beta$ satisfy the harmonic condition $1/\alpha + 1/\beta = 2$. Here the quantity $c ( M , N , \rho )$ is defined by
\begin{align}
c ( M , N , \rho ) 
:= \max_{ \rho = \sum_{k} u_{k} | u_{k} \rangle\langle u_{k} | } \max_{x,y}
\frac{ \Tr [ M_{x}^{\dagger} N_{y} | u_{k} \rangle\langle u_{k} |  ] }
{\| M_{x}^{1/2} | u_{k} \rangle  \| \cdot \| N_{y}^{1/2} | u_{k} \rangle \|}. 
\end{align}
\end{lem}

\noindent The method of proof employed Naimark's dilation theorem \cite{paulsen_2003} and Riesz theorem as expected. By using the properties of operator norm, that is $\| \bigcdot \| := \max \{ \| \bigcdot u \| \, | \, \| u \| = 1 \}$, lemma \ref{mixedmulemma} leads to the following entropic uncertainty relations with a state-independent bound $c ( M , N ) := \max_{x,y} \| M_{x}^{1/2} N_{y}^{1/2}\|$

\begin{cor}[Rastegin]\label{mixedmucorollary}
For probability vectors $\p$ and $\q$ obtained by measuring quantum state $\rho$ with respect to POVMs $M$ and $N$, their joint uncertainties in terms of $\mathrm{H}_{\alpha} (M) + \mathrm{H}_{\beta} (N)$ is therefore bounded by the maximum overlap $c ( M , N )$ of the form
\begin{align} 
\mathrm{H}_{\alpha} ( M )
+
\mathrm{H}_{\beta} ( N ) \geqslant 
-2\log c ( M , N ),
\end{align}
where $\alpha$ and $\beta$ satisfy the harmonic condition $1/\alpha + 1/\beta = 2$.
\end{cor}

There are two ways of proving Maassen-Uffink uncertainty relation for quantum processes. The first one is to apply Naimark's dilation theorem to the CJ matrix $J_{\mathrm{\Psi}}^{A B} $ with respect to the process $\mathrm{\Psi}$, followed by Riesz theorem. Another way is to use corollary \ref{mixedmucorollary} directly, which has been adopted here. 

For probability distribution $\p$ specified by the process-channel measurement $\mathcal{T}_{1}$, the probability associated with measurement outcome $x$, as shown in (\ref{pro}), is $p_{x} = \Tr [ E_{x}  J_{\mathrm{\Psi}}^{A B} ]$, and hence
\begin{align}
\frac{p_{x}}{d_{A}} = \Tr \left[ E_{x}  ~\rho_{\mathrm{\Psi}}^{A B} \right],
\end{align}
with $\rho_{\mathrm{\Psi}}^{A B} := J_{\mathrm{\Psi}}^{A B}/d_{A}$ being a bipartite quantum state in $D (A \otimes B)$, since $\rho_{\mathrm{\Psi}}^{A B} \geqslant 0$ (due to the CP of $\mathrm{\Psi}$) and $\Tr [\rho_{\mathrm{\Psi}}^{A B}] =1$ (due to the TP of $\mathrm{\Psi}$). Therefore, the probability distribution $\frac{1}{d_{A}} \p \oplus \frac{d_{A}-1}{d_{A}}$ can be seen as derived by performing POVM $\{ \tilde{E}_{x} \}_{x=1}^{m+1}$ to the state $\rho_{\mathrm{\Psi}}^{A B}$. Consider also the probability distribution $\frac{1}{d_{A}} \q \oplus \frac{d_{A}-1}{d_{A}}$ obtained by implementing POVM $\{ \tilde{F}_{y} \}_{y=1}^{n+1}$ to $\rho_{\mathrm{\Psi}}^{A B}$, then corollary \ref{mixedmucorollary} immediately implies that
\begin{align} 
\mathrm{H}_{\alpha} ( \tilde{E} )
+
\mathrm{H}_{\beta} ( \tilde{F} ) \geqslant 
-2\log c ( \mathcal{T}_{1} , \mathcal{T}_{2} ),
\end{align}
with $1/\alpha + 1/\beta = 2$. Written in full, that is
\begin{align} 
\mathrm{H}_{\alpha} ( \frac{1}{d_{A}} \p \oplus \frac{d_{A}-1}{d_{A}} )
+
\mathrm{H}_{\beta} ( \frac{1}{d_{A}} \q \oplus \frac{d_{A}-1}{d_{A}} ) \geqslant 
-2\log c ( \mathcal{T}_{1} , \mathcal{T}_{2} ),
\end{align}
as required.

%%%%%%%%%%%%%%%%%%%%%%%%%%%%%%%%%%%

\subsection{Proof of Theorem \ref{uuroplusthm}}

Our goal in this section is to prove (\ref{uuroplus}). Let us first consider the following question : for any semi-definite positive operator $W \in L (A \otimes B)$, what is the maximal value of $\Tr [ W J_{\mathrm{\Psi}}^{A B} ]$ for all quantum process? In particular, we are interested in
\begin{align}\label{primalsdp}
\max \quad & \Tr [ W J_{\mathrm{\Psi}}^{A B} ] \notag\\
\text{s.t.} \quad & \Tr_{B} J_{\mathrm{\Psi}}^{A B} = \mathds{1}^{A},\notag\\
& \quad ~ ~ J_{\mathrm{\Psi}}^{A B} \geqslant 0.
\end{align}
which is a semidefinite programming (SDP). The Lagrangian associated to the primal SDP in (\ref{primalsdp}) is given by :
\begin{align}
\mathscr{L}
&=
\Tr [ W J_{\mathrm{\Psi}}^{A B} ] 
+ 
\Tr_{A} [ X \left( \mathds{1}^{A} - \Tr_{B} J_{\mathrm{\Psi}}^{A B} \right) ] +
\Tr [ Y J_{\mathrm{\Psi}}^{A B} ]
\notag\\
&=
\Tr [ X ] 
+
\Tr [ \left( W + Y - X \otimes \mathds{1}^{A} \right) 
J_{\mathrm{\Psi}}^{A B} ],
\end{align}
where we have introduced dual variables, i.e. Lagrange multipliers, $X$, a Hermitian operator acting on Hilbert space $\mathcal{H}^{A}$, and $Y$, a semi-definite positive operator acting on Hilbert space $\mathcal{H}^{A} \otimes \mathcal{H}^{B}$, to ensure that the Lagrangian $\mathscr{L}$ is always greater than the objective function whenever the primal constraints are satisfied. Therefore, in this case, the dual SDP is obtained by minimizing over all dual variables :
\begin{align}\label{dualsdp}
\min \quad & \Tr [ X ] \notag\\
\text{s.t.} \quad & X \otimes \mathds{1}^{B} \geqslant W,
\end{align}
Here the strong duality holds since the primal SDP is finite and strictly feasible, which guarantees that the optimal value of dual coincides with the optimal value of the primal problem. Actually, the optimal value is related with the conditional min-entropy mentioned in our main text. We now move to the definition of conditional min-entropy \cite{renner2005security}, which is the main object of study in this section.

\begin{definition}[Min-entropy]
Let $\rho \in D (A\otimes B)$ be a bipartite quantum operator. The min-entropy of $A$ conditioned on $B$ is defined by
\begin{align}
\mathrm{H}_{\min} (A|B)_{\rho}
:=
- \inf_{\sigma} \mathrm{D}_{\infty} ( \, \rho \, \| \, \mathds{1}^{A} \otimes \sigma),
\end{align}
where the infimum ranges over all semidefinite positive operator $\sigma \in L(B)$, with
\begin{align}
\mathrm{D}_{\infty} ( \, \tau \, \| \, \eta \, ) 
:=
\inf \{ \, \lambda \in \mathds{R} \, | \, 2^{\lambda} \eta \geqslant \tau \, \}.
\end{align}
\end{definition}

\noindent Now it is clear from the context that the optimal value of (\ref{dualsdp}) equals to $2^{- \mathrm{H}_{\min} (B|A)_{W}}$, which is equivalent to say that for any quantum process $\mathrm{\Psi} : A \rightarrow B$, we have
\begin{align}\label{minentropy}
\max_{\mathrm{\Psi}} \Tr [ W J_{\mathrm{\Psi}}^{A B} ] 
=
2^{- \mathrm{H}_{\min} (B|A)_{W}}.
\end{align}

We now move on to discuss the sum of the first $k$ largest components of $\p \oplus \q$, i.e.
\begin{align}
\max_{ | R | + | S | = k } \max_{ \mathrm{\Psi} }
\left( \sum_{ x \in R } p_{x} + \sum_{ y \in S } q_{y} \right)
=
&\max_{ \mathcal{I}_{k} } \max_{ \mathrm{\Psi} }
\Tr [ \left( \sum_{ z \in \mathcal{I}_{k} } G_{z} \right) 
J_{\mathrm{\Psi}}^{A B} ] \notag\\
=
&\max_{ \mathcal{I}_{k} } \max_{ \mathrm{\Psi} }
\Tr [ G\left(\mathcal{I}_{k}\right) 
J_{\mathrm{\Psi}}^{A B} ] \notag\\
=
&\max_{ \mathcal{I}_{k} } 
2^{- \mathrm{H}_{\min} (B|A)_{G\left(\mathcal{I}_{k}\right)}} \notag\\
=
&s_{k}.
\end{align}
with $R \subset \left\{1, \ldots, n \right\}$, $S \subset \left\{1, \ldots, m \right\}$, and $|\bigcdot|$ stands for the cardinality of set $\bigcdot$. Here to arrive at the third line we used the result shown in (\ref{minentropy}), and the last line follows from the definition of $s_{k}$. Noticing now that when the the first $k$ largest components of $\p \oplus \q$ is upper-bounded by the quantity $s_{k}$, the vector $\p \oplus \q$ is thus majorized by $( s_{1} , s_{2} - s_{1} , s_{3} - s_{2} , \ldots, 0 )$. We finally remark that for the sum of the first $k$ largest components, $s_{k}$ is tight for all $k$, since there always exists a quantum process, which might not be unique, such that $\max_{ | R | + | S | = k } ( \sum_{ x \in R } p_{x} + \sum_{ y \in S } q_{y} ) = s_{k}$. Even though each $s_{k}$ is tight, their collection $\sss$ is not always guaranteed to be optimal. The optimal bound for $\p \oplus \q$ will be given in the next section by considering the lattice structure of majorization.

%%%%%%%%%%%%%%%%%%%%%%%%%%%%%%%%%%%

\subsection{Majorization Lattice}

In this section we turn our attention to the concept of lattice and employ majorization lattice to study the optimal bounds of UURs for quantum processes. For simplicity, all vectors considered in this section belongs to the set $\mathds{R}^{d}$. Let us start with the definition of {\it Lattice}, which is

\begin{definition}[Lattice] 
A quadruple $(S, \sqsubset, \wedge, \vee)$ is called lattice if $\ \sqsubset$ is a partial oder on the set $S$ such that for all $\p$, $\q \in S$ there exists a unique greatest lower bound (GLB) $\p \wedge \q$ and a unique least upper bound (LUB) $\p \vee \q$ satisfying 
\begin{align}
\x \sqsubset \p, \, \x \sqsubset \q 
&\Rightarrow 
\x \sqsubset \p \wedge \q, \notag\\
\p \sqsubset \y, \, \q \sqsubset \y
&\Rightarrow
\p \vee \q \sqsubset \y.
\end{align}
for each $\x$, $\y \in S$.
\end{definition}

\noindent A special class of lattices are those which have GLB and LUB for all their subsets, namely complete lattice

\begin{definition}[Complete Lattice]
A lattice $(S, \sqsubset, \wedge, \vee)$ is called complete, if for any nonempty subset $R \subset S$, it has a LUB, denoted by $\vee R$ and a GLB, denoted by $\wedge R$. More precisely, if $\x, \y \in S$ such that $\x \sqsubset R \sqsubset \y$, i.e. $\x \sqsubset \p \sqsubset \y$ for all $\p \in R$, we thus have $\x \sqsubset \wedge R$ and $\vee R \sqsubset \y$.
\end{definition}

Before interpreting the majorization lattice, let us first introduce some notations that will be used frequently in this section.
\begin{align}\label{sets}
\mathds{R}^{d}_{+} 
& := 
\{ \x \in \mathds{R}^{d} \, \| \, x_{k} \geqslant 0, \, \forall 1 \leqslant k \leqslant d \, \}
\notag\\
\mathds{R}^{d, \, \downarrow}_{+} 
& := 
\{ \x \in \mathds{R}^{d}_{+} \, \| \, x_{k} \geqslant x_{k+1}, \, \forall 1 \leqslant k \leqslant d-1 \, \}
\notag\\
\mathds{P}_{n}^{d} 
&:= 
\{ \x \in \mathds{R}^{d}_{+} \, \| \, \sum_{k} x_{k} = n \, \}
\notag\\
\mathds{P}_{n}^{d, \, \downarrow} 
&:= 
\mathds{P}_{n}^{d} \cap \mathds{R}^{d, \, \downarrow}_{+}
\end{align}
With these notations, we now introduce the relation between lattice and majorization, which was first established by the notion of weak majorization in Bapat's work \cite{BAPAT199159}.

\begin{definition}[Weak Majorization]
For $\x = (x_{k})_{k}$, $\y = (y_{k})_{k} \in \mathds{R}^{d}$, we say that $\x$ is weakly majorized by $\y$, denoted by $\x \prec_{\text{w}} \y$ if $\sum_{k=1}^{i} x_{k}^{\downarrow} \leqslant \sum_{k=1}^{i} y_{k}^{\downarrow}$ for all $1 \leqslant i \leqslant d$.
\end{definition} 

Due to the importance of majorization lattice, we will review historical developments of this topic briefly. Some useful results will also be given in this section. In 1991, during Bapat's investigations of the singular values of complex square matrices \cite{BAPAT199159}, the completeness of weak majorization on $\mathds{R}^{d, \, \downarrow}_{+}$ was obtained as a by-product.

\begin{lem}[Bapat]
Let $S \subset \mathds{R}^{d}_{+}$ be a nonempty set, then there exists a unique GLB, denoted by $\wedge S$, under weak majorization ``$\prec_{\text{w}}$''.
\end{lem}

\begin{lem}[Bapat]
Let $S \subset \mathds{R}^{d}_{+}$ be a bounded set, i.e. $\x \prec_{\text{w}} S \prec_{\text{w}} \y$ for some $\x$ and $\y \in \mathds{R}^{d}_{+}$, then there exists a unique LUB, denoted by $\vee S$, under weak majorization ``$\prec_{\text{w}}$''.
\end{lem}

\noindent Then, it can be shown that, for the set $\mathds{P}_{n}^{d, \, \downarrow}$, the quadruple $(\mathds{P}_{n}^{d, \, \downarrow} , \prec_{\text{w}}, \wedge, \vee)$ is bounded since 
\begin{align}
(n/d, \ldots, n/d) 
\prec_{\text{w}} 
\mathds{P}_{n}^{d, \, \downarrow} 
\prec_{\text{w}} 
(n, 0, \ldots, 0), 
\end{align}
which immediately implies that for any nonempty subset $S \subset \mathds{P}_{n}^{d, \, \downarrow} \subset \mathds{R}^{d}_{+}$, it is bounded and has unique GLB $\wedge S$ and LUB $\vee S$. Thus, $\mathds{P}_{n}^{d, \, \downarrow}$ is complete under ``$\prec_{\text{w}} $''.

\begin{cor}\label{weakmajorizationlattice}
The quadruple $(\mathds{P}_{n}^{d, \, \downarrow} , \prec_{\text{w}}, \wedge, \vee)$ forms a complete lattice.
\end{cor}

\noindent Here we would like to note that for the set $\mathds{P}_{n}^{d}$, weak majorization ``$\prec_{\text{w}}$'' is only a preorder, i.e. a binary relation that is both reflexive and transitive. However, ``$\prec_{\text{w}}$'' is not antisymmetric; that is we cannot obtain $\x = \y$ when $\x \prec_{\text{w}} \y$ and $\y \prec_{\text{w}} \x$ holds. For example, by taking $\x = (1,0)$ and $\y = (0,1) \in \mathds{P}_{1}^{d}$, we have
$(1,0) \prec_{\text{w}} (0,1)$ and $(0,1) \prec_{\text{w}} (1,0)$, but $(1,0) \neq (0,1)$. Accordingly, $(\mathds{P}_{n}^{d} , \prec_{\text{w}}, \wedge, \vee)$ is not even a lattice. Weak majorization ``$\prec_{\text{w}}$'' becomes a partial order when all the probability distribution vectors are arranged in non-increasing order, i.e. embedded into $\mathds{P}_{n}^{d, \, \downarrow}$. 

We now demonstrate that not only $(\mathds{P}_{n}^{d, \, \downarrow} , \prec_{\text{w}}, \wedge, \vee)$, but also $(\mathds{P}_{n}^{d, \, \downarrow} , \prec, \wedge, \vee)$ with majorization ``$\prec$'' forms a complete lattice. According to corollary \ref{weakmajorizationlattice}, there exist the GLB $\wedge S$ and LUB $\vee S$ for any nonempty subset $S$ of $\mathds{P}_{n}^{d, \, \downarrow}$, such that
\begin{align}
\wedge S \prec_{\text{w}} S \prec_{\text{w}} \vee S.
\end{align}
By considering the trivial bounds of subset $S \subset \mathds{P}_{n}^{d, \, \downarrow}$, i.e. $(n/d, \ldots, n/d)$, $(n,0, \ldots, 0) \in \mathds{P}_{n}^{d, \, \downarrow}$, which satisfies $(n/d, \ldots, n/d) \prec_{\text{w}} S \prec_{\text{w}} (n,0, \ldots, 0)$, we know that 
\begin{align}
(n/d, \ldots, n/d) \prec_{\text{w}} &\wedge S 
\prec_{\text{w}} (n,0, \ldots, 0), \notag\\
(n/d, \ldots, n/d) \prec_{\text{w}} &\vee S 
\prec_{\text{w}} (n,0, \ldots, 0),
\end{align}
which implies $\| \wedge S \, \|_{1} = \| \vee S \, \|_{1} = n$, and hence $\wedge S \prec S \prec \vee S$ holds for majorization ``$\prec$''. Till now we have shown that $\wedge S$ and $\vee S$ are lower bound and upper bound for $S$ respectively. Now it is time to prove that they are optimal under majorization. For any vector $\x \prec S$, it is also a lower bound for weak majorization, i.e. $\x \prec_{\text{w}} S$, and hence $\x \prec_{\text{w}} \wedge S$. Due to the fact that $\x \in \mathds{P}_{n}^{d, \, \downarrow}$, we have $\| \x \, \|_{1} = \| \wedge S \, \|_{1} = n$, and thus $\x \prec \wedge S$. Therefore $\wedge S$ is the GLB for $S$ under majorization. Similarly, we have that $\vee S$ is the LUB for $S$ under majorization, which leads to the following statement

\begin{cor}\label{majorizationlattice}
The quadruple $(\mathds{P}_{n}^{d, \, \downarrow} , \prec, \wedge, \vee)$ forms a complete lattice.
\end{cor}

\noindent A special class of Corollary \ref{majorizationlattice} is that $(\mathds{P}_{1}^{d, \, \downarrow} , \prec, \wedge, \vee)$ forms a complete lattice, i.e. the probability simplex in finite dimensional space with non-increasing order forms a complete lattice \cite{992785}. Moreover, this result has been used to derive the optimal common resource in majorization-based resource theories \cite{Bosyk_2019}, and optimal direct-sum UURs for quantum states \cite{yuan2019strong,doi:10.1002/andp.201900143} recently.

Now it is clear from the context that the optimal bound for $\p \otimes \q$ exists. Define the set $S_{\otimes}^{\text{pre}} := \{ \p \otimes \q \}$, where $\p$ and $\q$ are obtained by performing process-channel measurements $\mathcal{T}_{1}$ and $\mathcal{T}_{2}$ to a quantum process respectively. Then the set $S_{\otimes} := S_{\otimes}^{\text{pre}} \cap \mathds{P}_{1}^{d, \, \downarrow} \subset \mathds{P}_{1}^{d, \, \downarrow}$, and our corollary \ref{majorizationlattice} immediately implies the existence of $\wedge S_{\otimes}$ and $\vee S_{\otimes}$ under majorization. 
\begin{align}
\wedge S_{\otimes} 
\prec
\p \otimes \q 
\prec
\vee S_{\otimes},
\end{align}
Even though corollary \ref{majorizationlattice} ensures the existence of both the upper and lower bounds of $\p \otimes \q$, it does not teach us how to find them effectively. Note also that, the completeness of $(\mathds{P}_{1}^{d, \, \downarrow} , \prec, \wedge, \vee)$ cannot be applied to the direct-sum form straightway since $\p \oplus \q \notin \mathds{P}_{1}^{d, \, \downarrow}$. In this case, we can define the set $S_{\oplus}^{\text{pre}} := \{ \p \oplus \q \}$, and $S_{\oplus} := S_{\oplus}^{\text{pre}} \cap \mathds{P}_{2}^{d, \, \downarrow} \subset \mathds{P}_{2}^{d, \, \downarrow}$. The existence of the GLB $\wedge S_{\oplus}$ and LUB $\vee S_{\oplus}$ is guaranteed by corollary \ref{majorizationlattice}, which satisfies  
\begin{align}
\wedge S_{\oplus} 
\prec
\p \oplus \q 
\prec
\vee S_{\oplus},
\end{align}
with $\p$ and $\q$ obtained by performing process-channel measurements $\mathcal{T}_{1}$ and $\mathcal{T}_{2}$ to a quantum process respectively.

In order to find the optimal bounds for $S_{\oplus}$, an additional process, namely flatness process, is needed. In 2002, the lattice structure of majorization was revisited by Cicalese and Vaccaro in the study of its supermodularity and subadditivity properties \cite{992785}, and the well-known flatness process $\mathcal{F}$ was introduced.

\begin{definition}[Flatness Process]
\label{flatnessprocessdefinition}
Let $\x \in \mathds{R}^{d}_{+}$ be a vector, and $j$ be the smallest integer in $\left\{2, \ldots, d\right\}$ such that $x_{j}>x_{j-1}$, and $i$ be the greatest integer in $\left\{1, \ldots, j-1\right\}$ such that $x_{i-1} \geqslant (\sum_{k=i}^{j} x_{k})/(j-i+1):=a$. Define
% $\bm t := \left(T_{1}, \ldots, T_{n}\right)$ with
% $T_k = a$ for $k = i,\cdots,j$ and $T_k = S_k$ otherwise.
%
\begin{align}\label{eq; mj bound t}
\mathcal{F} (\x) := \left(x_{1}^{\prime}, \ldots, x_{n}^{\prime}\right) \, \text{with} \, 
x_{k}^{\prime} = 
     \begin{cases}
       a & \text{for}\quad k = i, \ldots, j \\
       x_{k} & \text{otherwise.} \\ 
     \end{cases}
\end{align}
\end{definition}

\noindent which satisfies the following lemma

\begin{lem}[Cicalese-Vaccaro]
\label{flatnessprocesslemma}
For any $\x \in \mathds{P}_{n}^{d}$, we have $\mathcal{F} (\x) \in \mathds{P}_{n}^{d, \, \downarrow}$, and $\sum_{i=1}^{k} x_{i} \leqslant \sum_{i=1}^{k} x_{i}^{\prime}$ for all $1 \leqslant k \leqslant d$. Moreover, for all $\y  \in \mathds{P}_{n}^{d, \, \downarrow}$, we have
\begin{align}
\sum_{i=1}^{k} x_{i} \leqslant \sum_{i=1}^{k} y_{i}, 
\quad \forall 1 \leqslant k \leqslant d \quad
\Rightarrow 
\quad
\mathcal{F} (\x) \prec \y.
\end{align}
\end{lem}

\noindent We stress here that the original statement of flatness process $\mathcal{F}$, including its definition and lemma \ref{flatnessprocesslemma}, introduced in \cite{992785} is only designed for the set $\mathds{P}_{1}^{d}$, i.e. probability simplex. However, its generalization for vectors in $\mathds{P}_{n}^{d}$, i.e. lemma \ref{flatnessprocesslemma}, is also valid. The corresponding proof was given in our recent work \cite{yuan2019strong}. 

All these properties mentioned above lead to a standard approach in finding the optimal bounds for a subset $S$ of $\mathds{P}_{n}^{d, \, \downarrow}$. Formally, let us consider $S \subset \mathds{P}_{n}^{d, \, \downarrow}$, and then there are two steps in constructing its GLB $\wedge S$ and LUB $\vee S$. The first step is to find the quantities $a_{k}$ and $b_{k}$, which are defined as
\begin{align}
a_{k} 
&:= 
\left( \min_{\x \in S} \sum_{i=1}^{k} x_{i} \right) -
\sum_{i=1}^{k-1} a_{i}, \notag\\
b_{k} 
&:= 
\left( \max_{\x \in S} \sum_{i=1}^{k} x_{i} \right) -
\sum_{i=1}^{k-1} b_{i},
\label{abbbdefinition}
\end{align}
for $1 \leqslant k \leqslant d$. It is immediate to observe that the vector $\ab_{S} := (a_{k})_{k} \in \mathds{P}_{n}^{d, \, \downarrow}$. On the other hand, the vector $\bb_{S} := (b_{k})_{k}$ might not always belongs to the set $\mathds{P}_{n}^{d, \, \downarrow}$. To give our reader some intuition, we recall the example constructed in \cite{992785}. 

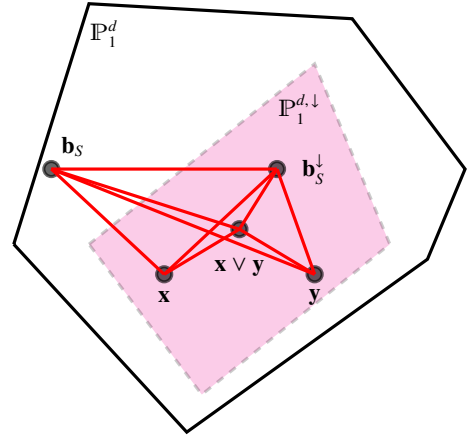
\begin{figure}[h]
\centering
\begin{tikzpicture}

\draw[very thick,dashed,-,fill=magenta,opacity=0.2] (-2,0) -- (1,2.4) -- (2,0) -- (-0.5,-2) -- (-2,0);
\draw[very thick,-] (-3,0) -- (-2,3.2) -- (1.5,3) -- (3,1) -- (2.5,-0.2) -- (-0.7,-2.5) -- (-3,0);

\draw[very thick,fill=black,opacity=0.6] (-1,-0.4) circle (0.1);
\node[] at (-1,-0.7) {$\x$};
\draw[very thick,fill=black,opacity=0.6] (1,-0.4) circle (0.1);
\node[] at (1,-0.7) {$\y$};
\draw[very thick,fill=black,opacity=0.6] (0,0.2) circle (0.1);
\node[] at (0,-0.3) {$\x \vee \y$};
\draw[very thick,red,-] (0,0.2) -- (-1,-0.4);
\draw[very thick,red,-] (0,0.2) -- (1,-0.4);
\draw[very thick,fill=black,opacity=0.6] (0.5,1) circle (0.1);
\node[] at (1,1) {$\bb_{S}^{\downarrow}$};
\draw[very thick,red,-] (0.5,1) -- (-1,-0.4);
\draw[very thick,red,-] (0.5,1) -- (1,-0.4);
\draw[very thick,red,-] (0.5,1) -- (0,0.2);
\node[] at (0.8,1.8) {$\mathds{P}_{1}^{d, \, \downarrow}$};
\node[] at (-1.8,2.8) {$\mathds{P}_{1}^{d}$};
\draw[very thick,fill=black,opacity=0.6] (-2.5,1) circle (0.1);
\node[] at (-2.2,1.3) {$\bb_{S}$};
\draw[very thick,red,-] (-2.5,1) -- (-1,-0.4);
\draw[very thick,red,-] (-2.5,1) -- (1,-0.4);
\draw[very thick,red,-] (-2.5,1) -- (0,0.2);
\draw[very thick,red,-] (-2.5,1) -- (0.5,1);

\end{tikzpicture}
\caption{(color online) Schematic illustration of the lattice structure exhibited in example \ref{example} excluding the GLB. Each point stands for an element, and the red line represents the binary relation ``$\prec$'' between elements. In this plot, a lower point is majorized by the higher point whenever they are connected with a red line. Obviously, here $\bb_{S}^{\downarrow} \prec \bb_{S}$ and $\bb_{S} \prec \bb_{S}^{\downarrow}$, but $\bb_{S} \neq \bb_{S}^{\downarrow}$.}
\label{lattice}
\end{figure}

\begin{example}\label{example}
Take $S = \{ \x, \y \}$ with
\begin{align}
\x &= (0.6, 0.15, 0.15, 0.1),\notag\\
\y &= (0.5, 0.25, 0.2, 0.05).
\end{align}
Then in this case $\bb_{S} = (0.6, 0.15, 0.2, 0.05)$, which does not belong to the set $\mathds{P}_{1}^{d, \, \downarrow}$ since $b_{2} = 0.15 < b_{3} = 0.2$. Actually, even though we rearrange the vector $\bb_{S}$ into non-increasing order $\bb_{S}^{\downarrow} = (0.6, 0.2, 0.15, 0.05)$, $\bb_{S}^{\downarrow}$ is not the optimal upper bound, i.e. $\bb_{S}^{\downarrow} \neq \vee S$, since in this case, 
\begin{align}
\vee S = \x \vee \y = \mathcal{F} (\bb_{S}) = (0.6, 0.175, 0.175, 0.05).
\end{align}
\end{example}

In general, the second step in constructing the optimal bounds for $S$ with majorization is to keep $\ab_{S}$ fixed and apply the flatness process $\mathcal{F}$ to $\bb_{S}$. Formally, our corollary \ref{majorizationlattice} and lemma \ref{flatnessprocesslemma} imply the optimality of $\ab_{S}$ and $\mathcal{F} (\bb_{S})$

\begin{cor}
\label{optimalmajorization}
For any nonempty subset $S \subset \mathds{P}_{n}^{d, \, \downarrow}$, its GLB $\wedge S$ and LUB $\vee S$ under majorization are given by
\begin{align}
\wedge S &= \ab_{S},\notag\\
\vee S &= \mathcal{F} (\bb_{S}),
\end{align}
with $\mathcal{F}$ stands for the flatness process defined in definition \ref{flatnessprocessdefinition}, and $\ab_{S}$, $\bb_{S}$ are defined in (\ref{abbbdefinition}).
\end{cor}

\begin{proof}
Here the existence of $\wedge S$ and $\vee S$ for $S$ are guaranteed by corollary \ref{majorizationlattice}. We first prove $\wedge S = \ab_{S}$. By hypothesis, for any vector $\cb \in \mathds{P}_{n}^{d, \, \downarrow}$ such that
\begin{align}
\cb \prec S,
\end{align}
we have
\begin{align}
\sum_{i=1}^{k} c_{i} 
\leqslant 
\sum_{i=1}^{k} a_{i},
\end{align}
for all $1\leqslant k \leqslant d$, and thus
\begin{align}
\cb \prec \ab_{S}.
\end{align}
In particular, by choosing $\cb$ as $\wedge S$, we obtain $\wedge S \prec \ab_{S}$. By the definition of $\ab_{S}$, we have $\ab_{S} \prec S$,
and hence $\ab_{S} \prec \wedge S$. Thus, $\wedge S = \ab_{S}$. The equation holds since majorization has the property of antisymmetricity on $\mathds{P}_{n}^{d, \, \downarrow}$, with both $\ab_{S}$ and $\wedge S$ belonging to the set $\mathds{P}_{n}^{d, \, \downarrow}$.

Next we move on to show $\vee S = \mathcal{F} (\bb_{S})$. By hypothesis, for any vector $\db \in \mathds{P}_{n}^{d, \, \downarrow}$ such that
\begin{align}
S \prec \db,
\end{align}
we have
\begin{align}
\sum_{i=1}^{k} b_{i} 
\leqslant 
\sum_{i=1}^{k} d_{i},
\end{align}
for all $1\leqslant k \leqslant d$. Now by using lemma \ref{flatnessprocesslemma} directly, we get 
\begin{align}
\mathcal{F} (\bb_{S}) \prec \db.
\end{align}
as expected. Note that $\mathcal{F} (\bb_{S}) \prec \bb_{S}$ since $\sum_{i=1}^{k} b_{i} \leqslant \sum_{i=1}^{k} b_{i}$. In particular, by choosing $\db$ as $\vee S$, we obtain $\mathcal{F} (\bb_{S}) \prec \vee S$. By using the fact that $\sum_{i=1}^{k} x_{i} \leqslant \sum_{i=1}^{k} x_{i}^{\prime}$ for all $1 \leqslant k \leqslant d$, and $\x \in S$, we have $S \prec \mathcal{F} (\bb_{S})$, and hence $\vee S \prec \mathcal{F} (\bb_{S})$. Thus, $\vee S = \mathcal{F} (\bb_{S})$. The equation holds since majorization has the property of antisymmetricity on $\mathds{P}_{n}^{d, \, \downarrow}$, with both $\mathcal{F} (\bb_{S})$ and $\vee S$ belonging to the set $\mathds{P}_{n}^{d, \, \downarrow}$.
\end{proof}

As an application of our corollary \ref{optimalmajorization}, take $S$ as $S_{\oplus}^{\downarrow} \subset \mathds{P}_{2}^{d, \, \downarrow}$, which immediately yields $\bb_{S_{\oplus}^{\downarrow}} = \sss$ defined in (\ref{uuroplus}) from our main text. Therefore, $\mathcal{F} (\sss) = \mathcal{F} (\bb_{S_{\oplus}^{\downarrow}}) = \vee S_{\oplus}^{\downarrow}$ is the optimal upper bound for UURs for all quantum processes in the form of direct-sum. Formally

\begin{cor}\label{uuropluscorop}
For probability vectors $\p$ and $\q$ obtained by measuring $\mathrm{\Psi}$ with respect to $\mathcal{T}_{1}:=( \rho^{R A}, M )$ and $\mathcal{T}_{2}:=( \sigma^{R A}, N )$, their joint uncertainties in terms of $\p \oplus \q$ is therefore bounded by a vector independent of quantum process $\mathrm{\Psi}$ of the form
\begin{align}\label{uuroplusop}
\p\oplus\q &\prec  
\mathcal{F} (\sss) 
=
\mathcal{F} \left( s_{1} , s_{2} - s_{1} , s_{3} - s_{2} , \ldots, 0 \right).
\end{align}
Here $\mathcal{F}$ is the flatness process defined in definition \ref{flatnessprocessdefinition}, $\mathcal{F} (\sss)$ is the optimal bound for $\p\oplus\q$, and each $s_{k}$ is a functional of the conditional min-entropy
\begin{align}
s_{k} := \max\limits_{\mathcal{I}_{k}} 2^{ -\mathrm{H}_{\min} (B|A)_{ G\left( \mathcal{I}_{k} \right) } },
\end{align}
where the maximum is over all subset $\mathcal{I}_{k}$, and the conditional min-entropy for $G\left( \mathcal{I}_{k} \right)$ is defined as
\begin{align}
\mathrm{H}_{\min} (B|A)_{G\left(\mathcal{I}_{k}\right)} := - \log \inf\limits_{X^{A} \geqslant 0} \left\{ \Tr \left(X^{A} \right) | X^{A} \otimes \mathds{1}^B \geqslant G\left(\mathcal{I}_{k}\right) \right\}.
\end{align}
\end{cor}

It turns out that not only the optimal upper bound $\mathcal{F} (\bb_{S_{\oplus}^{\downarrow}})$ of $S_{\oplus}^{\downarrow} \subset \mathds{P}_{2}^{d, \, \downarrow}$, i.e. direct-sum UURs for quantum processes, can be evaluated explicitly by the means of SDP and flatness process, but also the optimal lower bound $\ab_{S_{\oplus}^{\downarrow}}$ of the reverse direct-sum UURs for quantum processes.

%%%%%%%%%%%%%%%%%%%%%%%%%%%%%%%%%%%

\subsection{Proof of Theorem \ref{uurotimesthm}}

In this section we turn our attention back to the UURs for quantum processes in the form of direct-product. We first consider the sum of the first $k$ largest components of $\p \otimes \q$, i.e.

\begin{align}
\label{directproductinequality}
\max_{ T_{k} } \max_{ \mathrm{\Psi} }
\left( \sum_{ (x, y) \in T_{k} } p_{x} q_{y} \right)
\leqslant
&\max_{ | R | + | S | = k + 1 } \max_{ \mathrm{\Psi} }
\left( \frac{ \sum_{ x \in R } p_{x} + \sum_{ y \in S } q_{y} } {2} \right)^{2} \notag\\
=
&\max_{ \mathcal{I}_{k+1} } \max_{ \mathrm{\Psi} }
\left( \frac{\Tr [ G\left(\mathcal{I}_{k+1}\right) 
J_{\mathrm{\Psi}}^{A B} ]  }{2} \right)^{2} \notag\\
=
&\left(\frac{s_{k+1}}{2}\right)^{2} \notag\\
=
&t_{k},
\end{align}
where the outer maximum is over all subsets $T_{k} \subset [m] \times [n]$ such that $|T_{k}|=k$, with $[m] := \{ 1, \ldots, m \}$ and $[n] := \{ 1, \ldots, n \}$. Therefore $\ttt$ provides an upper bound of UURs for quantum processes, which completes the proof of our theorem \ref{uurotimesthm}. 

Moreover by definition of $S_{\otimes}^{\downarrow} \subset \mathds{P}_{1}^{d, \, \downarrow}$, and the iterated application of corollary \ref{optimalmajorization}, we have that
\begin{align}
\ab_{S_{\otimes}^{\downarrow}} 
\prec \p \otimes \q \prec 
\mathcal{F} (\bb_{S_{\otimes}^{\downarrow}})
 = 
 \vee S_{\otimes}^{\downarrow}.
\end{align}
It holds also that 
\begin{align}
\mathcal{F} (\bb_{S_{\otimes}^{\downarrow}})
\prec
\bb_{S_{\otimes}^{\downarrow}}.
\end{align}
Hence, the bounds for $\p \otimes \q$ can be ordered as
\begin{align}
\ab_{S_{\otimes}^{\downarrow}} 
\prec \p \otimes \q \prec 
\mathcal{F} (\bb_{S_{\otimes}^{\downarrow}})
\prec
\bb_{S_{\otimes}^{\downarrow}}.
\end{align}
From (\ref{directproductinequality}), it turns out that the bound $\bb_{S_{\otimes}^{\downarrow}}$ is majorized by the one constructed in our main text, that is $\bb_{S_{\otimes}^{\downarrow}} \prec \ttt$. Note that the quantity $\max_{ T_{k} } \max_{ \mathrm{\Psi} } ( \sum_{ (x, y) \in T_{k} } p_{x} q_{y} )$ is exactly the sum of the first $k$ largest components of $\bb_{S_{\otimes}^{\downarrow}}$, and usually 
\begin{align}
\bb_{S_{\otimes}^{\downarrow}} \neq \bb_{S_{\otimes}^{\downarrow}}^{\downarrow},
\end{align}
i.e. $\bb_{S_{\otimes}^{\downarrow}} \notin \mathds{P}_{1}^{d, \, \downarrow}$. Similarly, we have $\ttt \neq \ttt^{\downarrow}$ in general, and hence $\ttt \notin \mathds{P}_{1}^{d, \, \downarrow}$ does not hold in general.

It is interesting to identify the sum of the first $k$ largest components of $\mathcal{F} (\ttt)$. Let us denote the $i$-th element of $\mathcal{F} (\ttt)$ as $[\mathcal{F} (\ttt)]_{i}$. Then we have
\begin{align}
\max_{ T_{k} } \max_{ \mathrm{\Psi} }
\left( \sum_{ (x, y) \in T_{k} } p_{x} q_{y} \right)
\leqslant
t_{k}
\leqslant
\sum_{i=1}^{k} [\mathcal{F} (\ttt)]_{i},
\end{align}
and hence from lemma \ref{flatnessprocesslemma} we arrive at the following expression
\begin{align}
\ab_{S_{\otimes}^{\downarrow}} 
\prec \p \otimes \q \prec 
\mathcal{F} (\bb_{S_{\otimes}^{\downarrow}})
\prec
\mathcal{F} (\ttt)
\prec
\ttt. 
\end{align}
If the quantum processes considered here are state-preparation channel, then this chain of bounds makes an improvement over previous results of UURs introduced in \cite{PhysRevLett.111.230401} since $\mathcal{F} (\ttt) \prec \ttt$. As a by-product, the optimal bound $\ab_{S_{\otimes}^{\downarrow}} $ of the reverse direct-prooduct UURs for quantum processes is also given.

%%%%%%%%%%%%%%%%%%%%%%%%%%%%%%%%%%%

\subsection{Conjecture}

In this section we give a conjecture on the Shannon entropic uncertainty relation for quantum processes. In particular, given two process-channel measurements $\mathcal{T}_{1}$ and $\mathcal{T}_{2}$, their
overlaps are defined by $c_{xy} ( \mathcal{T}_{1} , \mathcal{T}_{2} ) := \Vert \tilde{E}_{x}^{1/2}  \tilde{F}_{y}^{1/2} \Vert$ with $1 \leqslant x \leqslant m+1$ and $1 \leqslant y \leqslant n+1$, and the entropic uncertainty relations in the form of $\mathrm{H}_{\alpha} (\mathcal{T}_{1}) + \mathrm{H}_{\beta} (\mathcal{T}_{2})$, with $1/\alpha + 1/\beta = 2$, is lower-bounded by $-2\log c ( \mathcal{T}_{1} , \mathcal{T}_{2} ) = -2\log \max_{x,y} c ( \mathcal{T}_{1} , \mathcal{T}_{2} )$, which is shown in our main text. This bound is tight for the case with $1/\alpha + 1/\beta = 2$. However, we do not know whether this is also tight for the case with $\alpha = \beta = 1$. 

As a matter of convenience, let us rearrange the overlaps between $\mathcal{T}_{1}$ and $\mathcal{T}_{2}$ in non-increasing order, and denote the $k$ largest overlap as $c_{k} ( \mathcal{T}_{1} , \mathcal{T}_{2} )$, then $c ( \mathcal{T}_{1} , \mathcal{T}_{2} ) = c_{1} ( \mathcal{T}_{1} , \mathcal{T}_{2} )$. Now we have a chain of overlaps
\begin{align}
c_{1} ( \mathcal{T}_{1} , \mathcal{T}_{2} )
\geqslant
c_{2} ( \mathcal{T}_{1} , \mathcal{T}_{2} )
\geqslant
\cdots
\geqslant
c_{(m+1)(n+1)} ( \mathcal{T}_{1} , \mathcal{T}_{2} ),
\end{align}
and we would like to know whether the Shannon entropic uncertainty relation can be further improve to 
\begin{align}
&\mathrm{H} (\mathcal{T}_{1}) + \mathrm{H} (\mathcal{T}_{2}) \notag\\
\geqslant
&-2\log c_{1} ( \mathcal{T}_{1} , \mathcal{T}_{2} )
+
\sum_{k} (2-s_{2k}) \log\frac{c_{k} ( \mathcal{T}_{1} , \mathcal{T}_{2} )}{c_{k+1} ( \mathcal{T}_{1} , \mathcal{T}_{2} )},
\end{align}
with $s_{k}$ is defined in (\ref{sk}) of our main text. In fact, when the object of our study is state-preparation channel, then the validity of above entropic uncertainty relation is proved by replace $c_{k} ( \mathcal{T}_{1} , \mathcal{T}_{2} )$ with $c_{k} ( M , N )$ in \cite{Xiao_2016}.

We will finish by expounding the motivations of this conjecture. Firstly, the process-independent bound depends only on the process-channel measurements $\mathcal{T}_{1}$ and $\mathcal{T}_{2}$, and hence quantify the intrinsic incompatibility between them. However, in the context of incompatibility, the process-independent bound is by no reason only dependent on the largest overlap between $\mathcal{T}_{1}$ and $\mathcal{T}_{2}$, but not all overlaps. The incompatibility between them should be completely characterized by the set of all overlaps. Secondly, it is worth noting that the bound of entropic uncertainty relation could be directly used to prove cryptography security \cite{RevModPhys.89.015002}.

%%%%%%%%%%%%%%%%%%%%%%%%%%%%%%%%%%%

\newpage
~

\end{document}